\documentclass[a4paper]{article}[12]

\usepackage[pages=all, color=black, position={current page.south}, placement=bottom, scale=1, opacity=1, vshift=5mm]{background}
\SetBgContents{
} 

\usepackage[margin=1in]{geometry} 
\usepackage{graphicx,subcaption,lipsum}
\usepackage{amsmath}
\usepackage{amsthm}
\usepackage{amssymb}
\usepackage[english]{babel}
\addto{\extrasenglish}{%
  
}
\usepackage{subcaption}
\usepackage{booktabs}
\usepackage{listings}
\usepackage{multirow}
\usepackage{url}

\usepackage{booktabs}
\usepackage{adjustbox}

\usepackage[utf8]{inputenc}
\usepackage{hyperref}
\hypersetup{
	unicode,
	colorlinks,
	breaklinks,
	urlcolor=blue, 
	linkcolor=blue, 
	pdfauthor={Author One, Author Two, Author Three},
	pdftitle={A simple article template},
	pdfsubject={A simple article template},
	pdfkeywords={article, template, simple},
	pdfproducer={LaTeX},
	pdfcreator={pdflatex}
}

\usepackage[numbers]{natbib}
\usepackage[nottoc]{tocbibind}

\usepackage{graphicx}
\theoremstyle{plain}
\newtheorem{theorem}{Theorem}

\theoremstyle{definition}

\usepackage{graphicx, color}
\usepackage{algorithm, algpseudocode} 
\usepackage{mathrsfs} 
\title{Modeling Transmission Dynamics of Tuberculosis: Parameter Estimation and Sensitivity Analysis Using Real-World Data}
\author{Moksina Seyid$^1$ \and Abdu Mohammed Seid$^1$ \and Yassin Tesfaw Abebe$^{1,2}$
}

\date{
$^1$Bahir Dar University, Ethiopia  \\ 
$^2$Mekdela Amba University, Ethiopia \\ 
}

\begin{document}
\maketitle
	
\begin{abstract}
Tuberculosis (TB) continues to pose a major public health challenge, particularly in high-burden regions such as Ethiopia, necessitating a more profound understanding of its transmission dynamics. In this study, we developed an SVEITRS compartmental model to investigate the transmission dynamics of TBs, utilizing real data from Ethiopia from 2011–2021. Model parameters were estimated via two methods: nonlinear least squares and maximum likelihood, with maximum likelihood providing more accurate and reliable results, as confirmed by a test case. The model's stability analysis indicated that there is a disease-free equilibrium in areas where the basic reproduction number ($\mathscr{R}_0$) is less than one. The results suggest that optimal conditions could lead to the elimination of TB. On the other hand, there is an endemic equilibrium in areas where $\mathscr{R}_0$ is greater than one, which means that the disease is still present. Sensitivity analysis revealed important factors affecting TB levels: higher natural death rates, vaccination rates, treatment rates, and disease-related death rates lower TB cases, whereas higher recruitment rates, contact rates, infection rates, and loss of vaccine protection increase its spread. These findings highlights to the necessity of enhancing vaccination, treatment, and recovery strategies while addressing drivers of transmission to achieve TB control in Ethiopia. This study provides useful advice for guiding TB control efforts and public health interventions in Ethiopia and similar regions.

\vspace{.25in}
\end{abstract}

\section{Introduction}
Mycobacterium tuberculosis is the bacterium that causes tuberculosis (TB), a chronic infectious disease. TB bacteria typically target the lungs (pulmonary TB), but they can also impact other body parts such as the brain, spine, bones, joints, central nervous system, circulatory system, and even the skin. Lung infections account for 75\% of TB cases. \cite{okolo2023mathematical}. When someone with tuberculosis coughs, sneezes, or speaks, bacteria are released into the air, which can potentially infect those in close proximity. This infectious disease spreads through the air and is still a leading cause of morbidity and death in many worldwide. Prior to the coronavirus disease 2019 (COVID-19) pandemic, TB was the most common infectious agent-related cause of death, surpassing HIV/AIDS \cite{world2021global,world2022global}. As a result, it was the second most common cause of infectious disease-related deaths, impacting numerous nations, particularly those in low-resource regions of Asia and Africa \cite{kasznia2023global}. According to estimates by \cite{world2024global}, TB killed almost twice as many people as HIV/AIDS did in 2023, and after COVID-19 temporarily surpassed it as the most common infectious cause of death for three years, it most likely reclaimed that title. An estimated 25\% of people worldwide are thought to have TB. \cite{world2023global}. Fever, weight loss, chest pain, blood in the cough, constant fatigue, night sweats, and appetite loss are signs of active tuberculosis \cite{diekmann2010construction}. According to Global TB 2023 report by the World Health Organization (WHO) \cite{world2023global}, the number of individuals with active TB increased from 10.0 million in 2020 to 10.3 million in 2021 and then to an estimated 10.6 million in 2022. Adults (aged $\geq$ 15 years) accounted for 88\% of these cases. Although tuberculosis affects people of all ages worldwide, adults are most likely to contract the disease in each nation. TB kills one person worldwide every 15 seconds, and the disease infects one person every second \cite{world2023global}. The WHO has therefore set a goal that has been embraced by all UN and WHO member states, requiring immediate action to end the global TB epidemic by 2030.

In Ethiopia, TB is one of the primary public health problems, killing more than 19 thousand people in 2022, as reported in the WHO 2023 report \cite{world2023global}. Poverty, undernutrition, overcrowded living conditions, and a lack of knowledge about the disease are known to increase the risk of bacteria1 spread and the risk of developing the disease.

Tuberculosis is treated by killing bacteria with antibiotics. The treatment usually lasts at least 6 months and sometimes longer, up to 24 months. Different antibiotics are used to increase effectiveness while preventing bacteria from becoming resistant to medicines. The most common medicines used to treat tuberculosis are isoniazid, rifampin, ethambutol, and pyrazinamide \cite{conradie2020treatment,dartois2022anti}. Bacillus Calmette-Gu\'erin (BCG) is the only licensed vaccine against TB. Some field trials of the vaccine have indicated protection rates as high as 70\% to 80\%, whereas others have shown it to be a completely ineffective vaccine for the prevention of TB \cite{castillo1996}. However, BCG has variable efficacy and cannot completely prevent TB infection and transmission \cite{plotkin2008vaccines, colditz1994efficacy, qu2021bcg}. Protecting against the infection itself is not the primary goal of the vaccine. Instead, the vaccine somewhat lowers the risk of contracting the disease's active form, but its effects are temporary and wear off over time. Therefore, the vaccine alone will not be enough to stop the disease's epidemic \cite{plotkin2008vaccines, colditz1994efficacy,qu2021bcg}. To efficiently control and prevent infectious diseases, one needs to be adequately informed about the mechanisms of the spread and transmission dynamics of such infectious diseases. 

Many studies have been performed on modeling the transmission dynamics of TBs to better understand transmission dynamics and propose corresponding possible control strategies, mainly following \citeauthor{waaler1962}'s 1962 study on the use of mathematical models for TB dynamics and  \citeauthor{castillo1996}'s 1996 study on modeling TB dynamics. Mathematical models in epidemiology provide convenient frameworks for improving our understanding of the patterns of infection and the effects of preventive measures on tuberculosis \citep{klotz2013forecast, mutai2021modeling}. In the last two decades, numerous mathematical models have been developed to study the transmission dynamics of TB across different regions in the world, including Ethiopia. The basic models, like susceptible-infectious-recovered (SIR) and its variants (SEIR, SVEIR), have been changed and expanded to include new factors such as treatment, vaccination, and disease-related deaths. This has given us a more complete picture of how TB spreads and how to prevent it \cite{adebiyi2016mathematical}. \citeauthor{nainggolan2013} \cite{nainggolan2013}, for example, presented an SVEIR model of tuberculosis transmission by specifically considering the total number of people who have recovered, whether naturally or as a result of vaccination. Using numerical examples, they demonstrated how vaccination can prevent the disease from spreading.

\citeauthor{khan2019} \cite{khan2019} and others studied on how tuberculosis spreads in Khyber Pakhtunkhwa, Pakistan. They used a deterministic model called Susceptible slow-Exposed fast-Exposed Infected Treated Recovered ($\mathrm{SE_1E_2ITR}$) to do this. They used local data from 2002-2017 for model parameter estimation and obtained a basic reproduction number, $\mathscr{R}_O= 1.38$, indicating the persistence of the disease in the region. However, many such studies have either relied on data from countries with more robust health infrastructure or have focused on theoretical applications rather than incorporating region-specific epidemiological data. Therefore, in many studies, especially in places with many TB cases, such as Ethiopia, models do not fully show how the disease spreads, what happens when people who have recovered lose their immunity over time, how important vaccinations are, or how treatment interventions help lower the disease burden. Some studies, such as those by \citeauthor{habteyohhanis2019analysis} and \cite{dauda2020analyzing}, provided us useful information about how TB spreads in certain Ethiopian communities. These studies focused on the basic reproduction number $\mathscr{R}_0$ and revealed patterns of endemic transmission. However, these studies did not incorporate factors such as the disease-induced death rate, immune loss in recovered individuals, or vaccination class, which are crucial for accurately capturing TB transmission dynamics throughout Ethiopia.

While the literature has made significant contributions, there remains a critical gap in modeling the full spectrum of TB dynamics in resource-limited such as Ethiopia, where healthcare challenges and socioeconomic factors complicate disease control. Many studies, including \citeauthor{andrawus2020mathematical} and \citeauthor{kereyu2021transmission}, have analyzed the role of treatment in reducing the TB burden but often overlooked the impact of vaccination or immune loss in recovered individuals. Early detection, treatment, and vaccination are some of the control strategies that have been explored in studies such as \citeauthor{haso2spread} and \citeauthor{sulayman2021sveire}. However, these studies do not consider all the important components, such as the treatment class or the resusceptibility of recovered individuals. Furthermore, many models, such as those by \citeauthor{mettle2020modelling}, do not explore the disease-induced death rate, a critical factor for regions with high mortality rates. Moreover, while some studies have attempted parameter estimation, they often fall to consider the unique temporal and socioeconomic factors that shape TB transmission in settings such as Ethiopia. These limitations underscore the need for more comprehensive models that incorporate these missing dynamics to offer more effective TB control strategies in high-burden areas such as Ethiopia.

To address these gaps, our study develops and analyzes a more comprehensive SVEITRS model that includes the loss of immunity in recovered individuals, treatment interventions, and the disease-induced death rate. This model not only adds the vaccination class, which was missing in previous models but also, introduces a more detailed representation of the various interventions needed to control TB spread effectively. For the estimation of the model parameters and the solution of the model, we used the corresponding Ethiopian data for the period 2011-2021 from \cite{worldbank_tb_ethiopia, worldbank2023, worldometers2024, tbdiah, tradingeconomics2023, worldbank_tuberculosis, tradingeconomics}. By utilizing both nonlinear least squares and maximum likelihood estimation methods for parameter estimation, we compare the accuracy of these techniques in capturing local TB dynamics. After the parameters are estimated, we use the model to predict the spread of tuberculosis and determine the likely number of infections and fatalities. Additionally, we use a synthetic example case to validate the model's performance in terms of parameter estimation and prediction accuracy, ensuring that the model can provide reliable insights for controlling TB in the Ethiopian context. This study thus fills a crucial gap in the literature by providing a more accurate, data-driven framework for TB transmission dynamics in Ethiopia, one that accounts for the full range of epidemiological factors and offers new insights into potential control strategies.

The remainder of this paper is arranged as follows. Section \ref{sec_two} presents the models used to study tuberculosis. We analyze the models in Section \ref{sec_three}. We estimate the values of the parameters in section \ref{sec_seven}. The results and discussion are presented in Section \ref{sec_eight}. The final conclusion and remarks are provided in section \ref{sec_nine}.

\section{Model Formulation }\label{sec_two}
In this section, we discuss the epidemiological model formulation for the transmission dynamics of TB in Ethiopia. We consider the following simplifying assumption in the construction of the model:
\begin{enumerate}
\item The total population is divided into six distinct compartments: susceptible $S(t)$, vaccinated $V(t)$, exposed $E(t)$, infectious $I(t)$, treated, $T(t)$ and recovered $R(t)$.
\item Susceptible individuals are people of all ages who can contract tuberculosis. 
The susceptible population is assumed to increase in size through the recruitment of individuals into the population at a constant rate $\Lambda$.
\item A fraction of the susceptible population is vaccinated, and individuals who have received vaccinations have partial immunity against TB. Some susceptible individuals who have been successfully vaccinated lose their vaccine-induced immunity and join the susceptible compartment again \cite{sulayman2021sveire}. Thus, the efficacy of the vaccine is assumed to be incomplete such that some portion of vaccinated individuals will be susceptible to bacteria at a rate $\theta$.
\item An infected individual has a latency period before becoming infectious. Thus, exposed individuals are TB-infected but not yet infectious and do not show symptoms of the disease. 
\item We assume that after interaction with an infected individual with a susceptible one with a contact rate $\beta$, a susceptible person moves into the exposed class.  Infectious individuals have the ability to spread the disease. 
\item Individuals in the population follow homogeneous mixing, where every susceptible individual has an equal likelihood of being infected when they come into sufficient contact with infectious individuals.
\item Treated individuals have received treatment for the disease, and recovered individuals are people who are cured of the disease.
\item The infectious class and treated class have an additional death rate due to the disease, and we denote these mortality rates as $\tau$ and $\omega$, respectively. 
\item Due to loss of immunity, we assume that recovered individuals move to the susceptible class at rate $\eta$. 
\item Treated individuals are recovered at a constant rate $\alpha$, and we assume that the natural death rate  is the same for all classes and is denoted by $\mu$.
\end{enumerate}

On the basis of the assumptions, descriptions, and interrelationships between the states and the parameters defined above, we illustrate the transmission dynamics of the TB model via a flowchart of the model dynamics among the different compartments, as shown in Figure \ref{fig:D_123}.
\begin{figure}[h]
\centering
\includegraphics[width=1\linewidth]{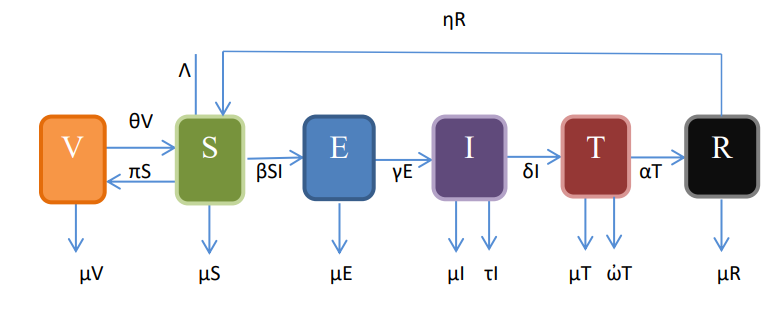}
\caption{Flowchart for the TB transmission model.}
\label{fig:D_123}
\end{figure}

According to the above hypotheses, the dynamics of disease within a specific population can be represented via the following set of nonlinear differential equations:
\begin{equation}\label{eq:4.1}
\begin{split}
\frac{dS}{dt} &= \Lambda + \eta R + \theta V - (\beta I + \pi + \mu)S \\
\frac{dV}{dt} &= \pi S - (\theta + \mu)V \\
\frac{dE}{dt} &= \beta SI - (\gamma + \mu)E \\
\frac{dI}{dt} &= \gamma E - (\delta + \tau + \mu)I \\
\frac{dT}{dt} &= \delta I - (\alpha + \omega + \mu)T \\
\frac{dR}{dt} &= \alpha T - (\eta + \mu)R
\end{split}
\end{equation}
with nonnegative initial conditions given by $S(0) = S_0 > 0$,~ $V(0) = V_0 \geq 0$,~ $E(0) = E_0 \geq 0$,~ $I(0) = I_0 \geq 0$,~ $T(0) = T_0 \geq 0$, \&~ $R(0) = R_0\geq 0$. All model parameters of the system in equation \eqref{eq:4.1} are assumed to be positive for all times $t > 0$ and are described in Table \ref{tab.1}.
\begin{table}[htbp]
\centering
\begin{adjustbox}{width=\textwidth}
\begin{tabular}{cllllll}
\hline
\textbf{Parameter} &&&&&& \textbf{Description}  \\
\hline
$\varLambda$ &&&&&& Recruitment rate of susceptible individuals due to birth. \\
\hline
$\beta$ &&&&&& Contact rate. \\
\hline
$\gamma$ &&&&&& The rate at which exposed individuals become infectious. \\
\hline
$\mu$ &&&&&& Natural death rates of all individuals. \\
\hline
$\pi$ &&&&&& Vaccination rate. \\
\hline
$\eta$ &&&&&& The rate at which the recovery loses its immunity.  \\
\hline
$\theta$ &&&&&& Loss rate of immunity by vaccinated individuals. \\
\hline
$\delta$ &&&&&& Treatment rate for infectious individuals. \\
\hline
$\tau$ &&&&&& TB-induced death rate for infected individuals.\\
\hline
$\alpha$ &&&&&& Recovery rate for treated individuals.  \\
\hline
$\omega$ &&&&&& TB-induced death rate for treated individuals. \\
\hline
\end{tabular}
\end{adjustbox}
\caption{Model parameters description}
\label{tab.1}
\end{table}
\section{Model Analysis}\label{sec_three}
\subsection*{Positivity of Solutions}
The proposed TB transmission dynamics model provided in equation \(\left( \ref{eq:4.1} \right)\) is epidemiologically meaningful if we can show that the solution of the system with nonnegative initial conditions will remain nonnegative for all $t > 0$.
\begin{theorem}
Let $S(0)\geq 0,~ V(0) \geq 0, ~ E(0) \geq  0,~ I(0) \geq  0,~ T(0) \geq 0$, and $R(0) \geq 0$, then the solutions $S(t), V(t), E(t), I(t), T(t), R(t)$ of the system of equations in \eqref{eq:4.1} are nonnegative for all $ t \geq 0$.
\end{theorem}
\begin{proof}
Consider $(S(t), V(t), E(t), I(t), T(t), R(t))$ to be the solution of the system of equations with the initial conditions $(S(0), V(0), E(0), I(0), T(0), R(0))$. Now, from the first equation of the system in \ref{eq:4.1}, we initially have		
\begin{equation*}
\displaystyle\frac{dS}{dt} = \Lambda + \eta R + \theta V-(\beta I + \pi + \mu)S  \geq - (\beta I + \pi + \mu)S.
\end{equation*}	
Using techniques of separation of variables, we obtain
\begin{equation}\label{eq.thm1}
\displaystyle\frac{dS}{S}\geq - (\beta I + \pi + \mu)dt.
\end{equation}
Integrating both sides of Equation \eqref{eq.thm1} and solving for $S(t)$ yields
\begin{equation*}
S(t)\geq S(0) \cdot e^{\displaystyle-\left(\beta I + \pi + \mu\right)t} > 0.
\end{equation*}
Similarly, for $V(t), E(t), I(t), T(t)$ and $R(t)$, we obtain
\begin{align*}
V(t) &\geq V(0) \cdot e^{\displaystyle -\left(\theta + \mu\right)t} \geq 0,~ E(t) \geq E(0) \cdot e^{\displaystyle -\left(\gamma +  \mu\right)t} \geq 0,\\
I(t) &\geq I(0) \cdot e^{\displaystyle -\left(\delta + \tau + \mu\right)t} \geq 0,~T(t) \geq T(0) \cdot e^{\displaystyle -\left(\alpha + \omega + \mu\right)t}\geq 0,~ \text{and}\\
R(t) &\geq R(0) \cdot e^{\displaystyle -\left(\eta +  \mu\right)t}\geq 0
\end{align*} \end{proof}
\subsection*{Invariant Region (Boundedness of the solution)}
\begin{theorem}
All the solutions $N(t) = S(t) +  V(t) +  E(t) + I(t) +T(t) + R(t)$ of the system of the model equation  \eqref{eq:4.1} are bounded.
\end{theorem}
\begin{proof}
To show that the population size of each compartment is bounded, we prefer to show that the total population size at time $t$ of the whole system $N(t)$ is bounded. The total population size $N$ at any time $t$  can be given by: 
\begin{equation*}
 N(t) = S(t) + V(t) + E(t) + I(t) + T(t) + R(t).
\end{equation*}
Differentiating with respect to time $t$ and using the system of equations of the model, we obtain: 
\begin{equation}\label{eq.thm2}
\displaystyle\frac{dN}{dt} = \displaystyle\frac{dS}{dt} + \displaystyle\frac{dV}{dt} + \displaystyle\frac{dE}{dt} + \displaystyle\frac{dI}{dt} + \displaystyle\frac{dT}{dt} +
\displaystyle\frac{dR}{dt}.
\end{equation}
Substituting the differential equation \eqref{eq:4.1} into Equation \eqref{eq.thm2}, we obtain
\begin{equation}\label{eq.thm3}
\displaystyle \frac{dN}{dt} = \Lambda - \mu N - \tau I - \omega T.
\end{equation}
Using the standard comparison theorem, if there is no disease-induced death rate of individuals in the population (i.e., $\tau$, $\omega$ = 0), Equation \eqref{eq.thm3} becomes
\begin{equation}\label{eq.thm4}
\displaystyle\frac{dN}{\Lambda - \mu N} \leq dt.
\end{equation}
By integrating both sides of Equation \eqref{eq.thm4} and solving it, we obtain
\begin{equation}\label{eq.thm5}
\Lambda - \mu N \geq N(0)e^{-\mu t}.
\end{equation}	
Since $\lim_{t \to \infty} \displaystyle e^{\displaystyle -\mu t} = 0$, Equation \eqref{eq.thm5} can be simplified as
\begin{equation*}
\Lambda - \mu N \geq  0 \Longrightarrow N \leq \displaystyle\frac{\varLambda}{\mu},
\end{equation*}
where $N$ is the total number of the human population. Therefore, the feasible solution set of equation \eqref{eq:4.1} is bounded and remains in the following region
\begin{equation*}
\Omega = \bigg\{(S,V,E,I,T,R) \in \mathbb{R}^6_+:0 \leq N \leq \displaystyle\frac{\varLambda}{\mu}\bigg\}.
\end{equation*}
\end{proof}

\subsection{Analysis of Disease-Free Equilibrium $(\mathscr{E}_0)$, and Basic Reproduction Number $(\mathscr{R}_0)$}\label{sec_four}
The mathematical expression $\mathscr{E}_0$ represents for the disease-free equilibrium point (DFEP), where there are no diseases in the human population. To obtain the disease-free equilibrium point, take the derivatives in the model equation \(\left( \ref{eq:4.1} \right)\) to zero and evaluate it at $E = I = T = R = 0$. Then, we are left with susceptible and vaccinated individuals from the model equation \(\left( \ref{eq:4.1} \right)\).
\begin{equation}
\label{eq:5.1}
\left\{
\begin{aligned}
\frac{dS}{dt} &= \Lambda + \eta R + \theta V - (\beta I + \pi + \mu)S \\
\frac{dV}{dt} &= \pi S - (\theta + \mu)V .
\end{aligned}
\right.
\end{equation} 
Thus, the disease-free equilibrium point (DFEP), which is given by
\begin{equation*}
\mathscr{E}_0 = \left( \displaystyle\frac{\varLambda(\theta + \mu)}{\mu(\theta + \pi + \mu)}, \frac{\varLambda\pi}{\mu(\theta + \pi + \mu)}, 0, 0, 0, 0 \right)
\end{equation*}
  
\subsubsection{Basic Reproductive Number}
The average number of secondary infections caused by a single infectious individual in a population that is completely susceptible during the mean infectious period is known as the basic reproduction number, denoted by $\mathscr{R}_0$ \cite{heffernan2005perspectives}.To obtain it, considering the model equations containing the infected compartments in general, i.e., $E$, $I$, and $T$, as shown in equation \eqref{eq:3} below, we used the next generation matrix method \cite{van2002reproduction}. 
\begin{equation}
\label{eq:3}
\left\{
\begin{aligned}
\frac{dE}{dt} &= \beta SI - (\gamma + \mu)E \\
\frac{dI}{dt} &= \gamma E  -(\delta +\tau + \mu)I \\
\frac{dT}{dt} &= \delta I  -(\alpha +\omega + \mu)T \\
\end{aligned}
\right.
\end{equation}
Splitting the right-hand side of equation \eqref{eq:3} as $\mathscr{F}-\mathscr{V}$, where $\mathscr{F}$ is the matrix of the new infection terms and $\mathscr{V}$ is the matrix of the transition terms in compartments $E ,I~ \& ~T$. Now, computing the Jacobian matrix $\mathscr{F}$ and $\mathscr{V}$ with respect to $E, I\ \&\ T$ and evaluating them at the disease-free equilibrium point (DFEP), $\mathscr{E}_0$, gives
\begin{center}
$F=\begin{bmatrix}
0 & \quad\displaystyle\frac{\varLambda\beta(\theta + \mu)}{\mu(\theta + \pi + \mu)}&\quad0 \\ 
0 & 0 &\quad 0\\
0 & 0 &\quad 0  
\end{bmatrix}$
\quad and \quad 
$V=\begin{bmatrix}
(\gamma + \mu) & 0& 0  \\
-\gamma & (\delta + \tau + \mu)&  0 \\
0 & - \delta&  (\alpha+\omega+\mu) 
\end{bmatrix}$.
\end{center}
Hence, following \cite{van2002reproduction}, the reproduction number of the model in equation \eqref{eq:4.1}, defined as the highest eigenvalue of $FV^{-1}$ given by $\mathscr{R}_0 = \rho(FV^{-1})$, is obtained as
\begin{equation}\label{eq.effective}
\mathscr{R}_v = \frac{\varLambda \beta \gamma (\theta + \mu)}{\mu (\delta + \tau + \mu)(\theta + \pi + \mu)(\gamma + \mu)}
\end{equation}
The threshold quantity $\mathscr{R}_V$ given in \ref{eq.effective} is known as the effective reproduction number, also called the control reproduction number. It calculates the typical number of new secondary infections that one infected person can spread over the course of an infection in a population that is totally vulnerable even with control measures in place \cite{ojo2023mathematical}. In this way, the number $\mathscr{R}_v$ in \ref{eq.effective} shows how many new TB cases one infected person is likely to cause in a population that is completely susceptible and has a vaccination program. A similar threshold quantity identified as the basic reproduction number can be obtained by setting the control measure and its parameters to zero (i.e., $\theta = \pi = 0$), such that 
\begin{equation}\label{eq.basic}
\mathscr{R}_0 = \frac{\varLambda \beta \gamma}{\mu (\delta + \tau + \mu)(\gamma + \mu)}
\end{equation}
Notably, the basic reproduction number $\mathscr{R}_0$ measures the average number of secondary cases produced by a single infected individual in a population that is completely susceptible, in a population without any control measure (see, for example, \cite{sulayman2021sveire}, \cite{ojo2023mathematical}).

\subsection{Stability Analysis}\label{sec_five}
Stability analysis can predict the long-term behavior of the model solution. We employ two types of stability analysis: local and global stability. Local stability is concerned with the behavior of the solution near the equilibrium points, whereas global stability can describe the behavior of the solution in the whole domain.

\subsubsection{Local Stability of the Disease-Free Equilibrium Point}
The DFEP $(\mathscr{E}_0)$ is said to be locally asymptotically stable if the real parts of the eigenvalue are all negative and unstable if at least one of the real parts of the eigenvalue is positive. If $\mathscr{R}_v< 1$, then the disease-free equilibrium point is locally asymptotically stable, i.e., no tuberculosis epidemic can develop in the population. If $\mathscr{R}_v> 1$, then the disease-free equilibrium point is unstable.

\begin{theorem}
The disease-free equilibrium point for model equation \eqref{eq:4.1} is locally asymptotically stable if $\mathscr{R}_v< 1$ and unstable if $\mathscr{R}_v> 1$.
\end{theorem}
\begin{proof}
The Jacobian matrix of the model equation \eqref{eq:4.1} at the disease free equilibrium point $\mathscr{E}_0$ is given by 
$$J(\mathscr{E}_0)=\begin{bmatrix} -(\pi + \mu) & \theta &  0 & \frac{-\varLambda \beta(\theta + \mu)}{\mu(\theta + \pi + \mu)} & 0 &  \eta\\ 
\pi &  -(\theta + \mu) &  0 &\quad  0 &  0 & 0\\ 
0  &  0 &  -(\gamma + \mu) &   \frac{\varLambda \beta(\theta + \mu)}{\mu(\theta + \pi + \mu)} &  0 & 0\\  
0 &  0 &  \gamma &  -(\delta + \tau + \mu) & 0 &  0\\ 
0 &  0 &  0 &  \delta &  -(\alpha + \omega + \mu) & 0\\
0 &  0 &  0 & 0 &  \alpha & -(\eta + \mu)\\ 
\end{bmatrix}.$$ 
	
The eigenvalues of $J(\mathscr{E}_0)$ can be obtained from $|J(\mathscr{E}_0) - \lambda I_6| = 0$, where $I_6$ is the identity matrix of order 6. Then we have
	
\begin{flalign}\label{eq.jacobian}
\begin{bmatrix} 
-(\pi\! +\! \mu\! +\! \lambda) & \theta &  0 & \frac{-\varLambda \beta(\theta\! +\! \mu\! +\! \lambda)}{\mu(\theta\! +\! \pi\! +\! \mu)} & 0 &  \eta\\
\pi &  -(\theta\! +\! \mu\! +\! \lambda) &  0 &\quad  0 &  0 & 0\\ 
0  &  0 &  -(\gamma\! +\! \mu\! +\! \lambda) &   \frac{\varLambda \beta(\theta\! +\! \mu)}{\mu(\theta\! +\! \pi\! +\! \mu)} &  0 & 0\\  
0 &  0 &  \gamma &  -(\delta\! +\! \tau\! +\! \mu \!+\! \lambda) & 0 &  0\\ 
0 &  0 &  0 &  \delta &  -(\alpha\! +\! \omega\! +\! \mu\! +\! \lambda) & 0\\
0 &  0 &  0 & 0 &  \alpha & -(\eta\! +\! \mu\! +\! \lambda)\\ 
\end{bmatrix}\!=\!0
\end{flalign}
	
From the matrix in Equation \eqref{eq.jacobian}, we obtain the characteristic equation given by
\begin{equation*}
\big[(\pi +\mu + \lambda)(\theta + \mu +\lambda) - \theta\pi\big]\big[(\eta+ \mu+ \lambda)(\alpha + \omega + \mu +\lambda)\big] \bigg[(\gamma +\mu + \lambda)(\delta + \tau + \mu + \lambda)-\displaystyle\frac{\varLambda \beta\gamma(\theta + \mu)}{\mu(\theta + \pi + \mu)}\bigg] = 0.
\end{equation*}
Obviously, $\lambda_1$, $\lambda_2$, $\lambda_3$, and $\lambda_4$ are negative, and for stability, all the real parts of the eigenvalues must be negative, which means 
\begin{equation*}
(\gamma +\mu)(\delta + \tau + \mu)-\displaystyle\frac{\varLambda \beta\gamma(\theta + \mu)}{\mu(\theta + \pi + \mu)} < 0 \Longrightarrow  \frac{\varLambda \beta\gamma(\theta + \mu)}{\mu(\theta + \pi + \mu)(\gamma +\mu)(\delta + \tau + \mu)} < 1.
\end{equation*}
Hence, $\displaystyle\mathscr{R}_v < 1$, and therefore, the disease-free equilibrium point $\mathscr{E}_0$ is locally asymptotically stable if $\mathscr{R}_v<1$.
\end{proof}

\subsubsection{Global Stability of the Disease-Free Equilibrium Point}
This section examines the global stability of the equilibrium points of the ordinary differential equation (ODE) in \eqref{eq:4.1}. We analyzed the global stability of the disease-free equilibrium point via the Castillo-Chavez approach \citep{castillo2002mathematical}. The Castillo-Chavez approach claims globally asymptotic stability of the disease-free equilibrium point if two specific conditions (specified below) are met when $\mathscr{R}_v < 1$. First, Equation \eqref{eq:4.1} can be restated in the following manner:
\begin{equation}\label{eq.global1}
\begin{cases}
\frac{dX}{dt} = F(X,Z) \\
\frac{dZ}{dt} = G(X,Z); \quad G(X,0) = 0
\end{cases}
\end{equation}
where $X =(S, V, R) \in \mathbb{R}^3_+$ represents the class of uninfected individuals and where $Z = (E, I, T)\in \mathbb{R}^3_+$ represents the class of infected individuals. Thus, $\mathscr{E}_0 = (x^*, 0)$ denotes the disease-free equilibrium of the system in equation \ref{eq.global1}. The specific conditions, I and II, are defined below and must be satisfied to guarantee local asymptotic stability:
\begin{itemize}
\item [\textbf{I}.] For $\frac{dX}{dt} = F(X, 0), \mathscr{E}_0$ is globally asymptotically stable.
\item [\textbf{II}.] $G(X, Z) = AZ - \hat{G}(X, Z),  \hat{G}(X, Z) \geq 0$ for all $(X , Z)$ $\in \mathbb{R} $, where $A = D_Z G(\mathscr{E}_0)$ is an M-matrix (the off-diagonal elements $A$ are nonnegative) and $\mathbb{R}$ is the region where the model makes biological sense. If System \ref{eq.global1} satisfies the above two conditions, then the following theorem holds \cite{castillo2002mathematical}:
\end{itemize}

\begin{theorem}[\textbf{Castillo-Chavez Theorem}]\label{CC}
The fixed point $\mathscr{E}_0 = (x^*, 0)$ is a globally asymptotic stable equilibrium of \ref{eq.global1}, provided that $\mathscr{R}_0<1$ under assumptions I and II are satisfied.
\end{theorem}

Now, an application of theorem \ref{CC} for System \eqref{eq:4.1} yields:
\begin{equation*}
X =(S, V, R)^T \in \mathbb{R}^3_+, Z = (E, I,  T)^T\in \mathbb{R}^3_+, \mathscr{E}_0 = (x^*, 0) = (S_0, V_0, 0).
\end{equation*}
As $t\rightarrow\infty$,\ $(S(t), V(t), R(t)) \rightarrow$ $(\frac{\Lambda(\theta + \mu)}{\mu(\theta + \pi + \mu)}$, $ \frac{\Lambda\pi}{\mu(\theta + \pi + \mu)}, 0) = \mathscr{E}_0 = (x^*, 0)$ is globally asymptotically stable. Moreover,

\begin{center}
$A= \begin{pmatrix}  
-(\gamma + \mu) &\frac{\Lambda\beta(\theta + \mu)}{\mu(\theta + \pi + \mu)} & 0 \\ 
\gamma & -(\delta + \tau + \mu) & 0 \\
0 & \delta & -(\alpha + \omega + \mu)
\end{pmatrix}\quad\&\quad \hat{G}(X, Z) = \begin{pmatrix}  
\beta I[S_0-S] \\ 
0 \\
0
\end{pmatrix}$	
\end{center}

For System \eqref{eq:4.1}, $S_0=\frac{\Lambda(\theta+\mu)}{\mu(\theta+\pi+\mu)}$, since $\frac{\Lambda}{\mu}$ represents the upper bound of the total population. Hence, $S_0 > S$, which leads to $\hat{G}(X, Z) \geq 0$ for all $(X,Z) \in \mathbb{R}^6_+$. Thus, the assumptions, \textbf{I} and \textbf{II}, of the Castillo-Chavez theorem are satisfied for System \eqref{eq:4.1}, hence, the theorem follows.

\subsection{Endemic Equilibrium Point $(\mathscr{E}^*)$, and Sensitivity Analysis}\label{sec_six}
The endemic equilibrium point (EEP) $\mathscr{E^*}=(S^*, V^*, E^*, I^*, T^*, R^*)$ is the point at which the disease remains alive in the population when at least one of the infected classes of the model is nonzero. To obtain the endemic equilibrium point, we equate all expressions of model equations \(\left( \ref{eq:4.1} \right)\) of the derivative to zero. From the system of equations \(\left( \ref{eq:4.1} \right)\), we get the following simplified endemic equilibrium point:
 \begin{align*}
S^* &= \frac{\varLambda(\theta + \mu)}{ \mu(\theta + \tau + \mu)\mathscr R_0},~ V^* = \frac{\varLambda\pi}{\mu(\theta + \pi + \mu)\mathscr R_0} \\
E^* &= \frac{\mu(\eta + \mu)(\alpha +\omega +\mu)(\gamma + \mu)(\theta + \pi +\mu)(\delta + \tau + \mu)^2(\mathscr{R}_0 -1)}{\beta\gamma(\theta + \mu)[(\delta + \tau + \mu)(\alpha + \omega + \mu)(\gamma + \mu)(\eta + \mu) - \eta\delta\gamma\alpha]},\\ 
I^* &= \frac{\mu(\eta + \mu)(\alpha +\omega +\mu)(\delta + \tau + \mu)(\gamma + \mu)(\theta + \pi +\mu)(\mathscr{R}_0 -1)}{\beta(\theta + \mu)[(\delta + \tau + \mu)(\alpha + \omega + \mu)(\gamma + \mu)(\eta + \mu) - \eta\delta\gamma\alpha]}, \\
T^* &= \frac{\delta\mu(\eta + \mu)(\delta + \tau + \mu)(\gamma + \mu)(\theta + \pi +\mu)(\mathscr{R}_0 -1)}{\beta(\theta + \mu)[(\delta + \tau + \mu)(\alpha + \omega + \mu)(\gamma + \mu)(\eta + \mu) - \eta\delta\gamma\alpha]},~ \text{and}\\ 
R^* &= \frac{\alpha\delta\mu(\delta + \tau + \mu)(\gamma + \mu)(\theta + \pi +\mu)(\mathscr{R}_0 -1)}{\beta(\theta + \mu)[(\delta + \tau + \mu)(\alpha + \omega + \mu)(\gamma + \mu)(\eta + \mu) - \eta\delta\gamma\alpha]}.
\end{align*}

\subsubsection{Local Stability of the Endemic Equilibrium Point}
\begin{theorem}
The endemic equilibrium point of the system of equations in  \eqref{eq:4.1} is locally asymptotically stable if $\mathscr{R}_0>1$ and unstable if $\mathscr{R}_0<1$.
\end{theorem}
\begin{proof}
The stability of the endemic equilibrium point is then determined on the bases of the sign of the eigenvalues of the Jacobian matrix obtained as: 
$$J(\mathscr{E^*})=\begin{bmatrix}
-( \beta I^* +\pi + \mu) & \theta & 0 & -\beta S^*& 0& \eta\\ 
\pi & -(\theta + \mu) & 0 & 0 & 0&0\\ 
\beta I^* & 0 & -(\gamma + \mu) &  \beta S^* & 0&0\\  
0 & 0 & \gamma & -(\delta + \tau + \mu) & 0&0\\ 
0 & 0 & 0 & \delta & -(\alpha + \omega + \mu)&0\\
0 & 0 & 0 & 0 &\alpha&-(\eta + \mu)
\end{bmatrix}.$$
To find the eigenvalues, we have
\begin{equation}\label{eq.locstable1}
|J (\mathscr{E^*}) - \lambda I_6|  = 0
\end{equation}
Let	$a = (\beta I^* + \pi + \mu), b = \beta I^*,  c = (\theta + \mu),  d = (\gamma + \mu), e = \beta S^*, f = (\delta + \tau + \mu), g = (\eta +\mu)$ and $h = (\alpha + \omega + \mu)$. Then the characteristic polynomial representation of equation \eqref{eq.locstable1} is given by:
\begin{equation*}
a_6 \lambda^6 + a_5\lambda^5 + a_4\lambda^4 + a_3\lambda^3 + a_2\lambda^2 + a_1\lambda^1 + a_0\lambda^0 = 0
\end{equation*}
 where
\begin{align*}
a_0 &= acdfgh + egh\pi\theta\gamma - acegh\gamma - bcegh\gamma - dfgh\pi\theta - bc\alpha\eta\gamma\delta,\\
a_1 &= acfgh + acdfg + acdfh + adfgh +cdfgh + acdgh + eg\pi\theta\gamma + eh\pi\theta\gamma - aceg\gamma - aceh\gamma- aegh\gamma - cegh\gamma -\\ 
& dfg\pi\theta - dfh\pi\theta - dgh\pi\theta - bceg\gamma - bceh\gamma - begh\gamma - fh\pi\theta - b\alpha\gamma\eta\delta, \\
a_2 &=acfg + acfh + afgh + cfgh + acgh + acdf + adfg + adfh + cdfg + cdfh + acdg + acdh + adgh + cdgh + \\ 
& e\pi\theta\gamma - ace\gamma - aeg\gamma - aeh\gamma - ceg\gamma - ceh\gamma- egh\gamma- bce\gamma -beg\gamma - beh\gamma - df\pi\theta - dg\pi\theta - dh\pi\theta - \\ 
& fg\pi\theta - fh\pi\theta - gh\pi\theta, \\
a_3 &= acf + afg + afh + cfg  + cfh + fgh + acg + ach + agh + cgh + adf + dfh + acd + adg + adh + cdg + cdh + \\ 
& dgh - ae\gamma - ce\gamma -be\gamma - eg\gamma- eh\gamma - eg\gamma - d\pi \theta - f\pi\theta - g\pi\theta - h\pi\theta, \\
a_4 &= af+ cf + fg + fh + ac + ag + ah +cg +ch +gh + df + ad + cd + dg +dh - e\gamma - \pi\theta,\\
a_5 &= (a + c + d + f + g +h),~ \text{and}~ a_6 = 1.
\end{align*}
After substitution and simplification by the Routh-Hurwitz stability criterion \citep{rcet}, the endemic equilibrium point $\mathscr{E}^*$ is locally asymptotically stable if $\mathscr{R}_0>1$.
\end{proof}

\subsubsection{Global Stability of Endemic Equilibrium Point}
\begin{theorem}
If $\mathscr{R}_0>1$, the endemic equilibrium $\mathscr{E}^*$ of the system of equations in  \eqref{eq:4.1} is globally asymptotically stable, whereas if $\mathscr{R}_0<1$, it is unstable.
\end{theorem}
\begin{proof}
To prove the global asymptotic stability at the endemic equilibrium point, we use the Lyapunov function $\mathcal{V}$ \cite{shuai2013global}, which is defined and derived as follows:: 
\begin{align*}
\mathcal{V}(S^*, V^*, E^*, I^*, T^*, R^*) &= \bigg(S - S^* - S^*\displaystyle\ln \frac{S^*}{S}\bigg) + \bigg(V - V^* - V^*\displaystyle\ln \frac{V^*}{V}\bigg) + \bigg(E - E^* - E^*\displaystyle\ln \frac{E^*}{E}\bigg) +\\  
& \bigg(I - I^* - I^*\displaystyle\ln \frac{I^*}{I}\bigg) + \bigg(T - T^* - T^*\displaystyle\ln \frac{T^*}{T}\bigg) + \bigg(R - R^* - R^*\displaystyle\ln \frac{R^*}{R}\bigg).
\end{align*}
By calculating the derivatives $\mathcal{V}(S^*, V^*, E^*, I^*, T^*, R^*)$ with respect to $t$, we obtain
\begin{align*}
\displaystyle\frac{d\mathcal{V}}{dt} &= \bigg(\displaystyle\frac{S-S^*}{S}\bigg)\displaystyle\frac{dS}{dt} +  \bigg(\displaystyle\frac{V - V^*}{V}\bigg)\displaystyle\frac{dV}{dt} +  \bigg(\displaystyle\frac{E - E^*}{E}\bigg)\frac{dE}{dt} + \\
& \bigg(\displaystyle\frac{I - I^*}{I}\bigg)\displaystyle\frac{dI}{dt} +  \bigg(\displaystyle\frac{T - T^*}{T}\bigg)\displaystyle\frac{dT}{dt} + \bigg(\displaystyle\frac{R - R^*}{R}\bigg)\displaystyle\frac{dR}{dt}	
\end{align*}
By direct substitution of the model equation \ref{eq:4.1}, we have
\begin{align*}
\displaystyle\frac{d\mathcal{V}}{dt} &= \bigg(\displaystyle\frac{S-S^*)}{S}\bigg)(\Lambda + \eta R + \theta V- (\beta I + \pi + \mu)S) + \bigg(\displaystyle\frac{V- V^*}{V}\bigg)(\pi S - (\theta + \mu)V) +\\
& \bigg(\displaystyle\frac{E-E^*}{E}\bigg)(\beta SI - (\gamma + \mu)E) +\bigg(\displaystyle\frac{I-I^*}{I}\bigg)(\gamma E -(\delta +\tau + \mu)I) +\\
& \bigg(\displaystyle\frac{T-T^*}{T}\bigg)(\delta I -(\alpha + \omega +\mu)T) +\bigg(\displaystyle\frac{R-R^*}{R}\bigg)(\alpha T - (\eta + \mu)R).
\end{align*}
Then by simplifying the above equation, we have
\begin{align*}
\displaystyle\frac{d\mathcal{V}}{dt} &= \Lambda + \beta IS^* + \pi S^* + \mu S^* + \theta V^* + \mu V^* + \gamma E^* + \mu E^* +\delta I^* +\tau I^* + \mu I^* + \alpha T^* + \\
& \omega T^* + \mu T^* + \eta R^* + \mu R^* - \mu S - \Lambda\displaystyle\frac{S^*}{S} - \eta R\frac{S^*}{S} - \theta V\frac{S^*}{S} - \mu V -\pi S\frac{V^*}{V}  - \mu E - \\
& \beta SI \displaystyle\frac{E^*}{E}- \tau I - \mu I - \gamma E \displaystyle\frac{I^*}{I} -\omega T -\mu T - \delta I \frac{T^*}{T}- \mu R - \delta I\frac{R^*}{R}.
\end{align*}
Now, let us take $X$ as the positive term and $Y$ as the negative term of the equation above such that
$$\displaystyle\frac{dV}{dt} = X - Y.$$
Thus, if $X < Y$, then $\displaystyle\frac{d\mathcal{V}}{dt}< 0$; note that $S = S^*, V = V^*, E = E^*, I = I^*, T = T^*, R = R^*$ if and only if  $\frac{d\mathcal{V}}{dt}= 0$. Therefore, Lasalle’s invariant principle \cite{lasalle1960some} indicates that the endemic equilibrium $\mathscr{E^*}$ is globally asymptotically stable if $X < Y$.
\end{proof}

\subsection{Sensitivity Analysis}
We aim to ascertain the sensitivity of each parameter to the effective reproduction number of the model. We apply the following sensitivity index formula:
\begin{equation}\label{eq.SA_1}
\displaystyle\varUpsilon^{\mathscr{R}_v}_{mi} = \displaystyle\frac{\partial \mathscr{R}_v }{\partial mi} \times \displaystyle\frac{mi}{\mathscr{R}_0}	
\end{equation}
where $mi$ are the system’s parameters in the quantity $\mathscr{R}_v$. We recall that the effective reproduction number $\mathscr{R}_v$ is given by:
\begin{equation}\label{basic_repro}
\mathscr{R}_v = \displaystyle\frac{\varLambda \beta \gamma(\theta + \mu)}{\mu(\theta + \pi + \mu)(\delta + \tau + \mu)(\gamma + \mu)}.
\end{equation}
Then, via Equation \eqref{basic_repro}, direct computation yields

\begin{align*}
\displaystyle \varUpsilon^{\mathscr{R}_v}_{\varLambda} &= 1 > 0,~\displaystyle \varUpsilon^{\mathscr{R}_v}_{\beta}  = 1 > 0,~ \displaystyle\varUpsilon^{\mathscr{R}_v}_{\gamma} = \displaystyle\frac{\mu}{(\gamma + \mu)}  > 0,~
\displaystyle\varUpsilon^{\mathscr{R}_v}_{\theta} = \displaystyle\frac{\pi}{(\theta + \pi + \mu)(\theta + \mu)}  > 0,\\
\displaystyle \varUpsilon^{\mathscr{R}_v}_{\mu} &= -[\displaystyle\frac{\theta}{(\theta + \mu)} + \frac{\mu}{(\theta + \pi + \mu)} + \frac{\mu}{(\delta + \tau + \mu)} + \frac{\mu}{(\gamma + \mu)}] < 0,\\
\displaystyle\varUpsilon^{\mathscr{R}_v}_{\pi} &=  \displaystyle\frac{-\pi}{(\theta + \pi + \mu)} < 0,~
\displaystyle\varUpsilon^{\mathscr{R}_v}_{\delta} = \displaystyle\frac{-\delta}{(\delta + \tau + \mu)} < 0,~\text{and}~
\displaystyle\varUpsilon^{\mathscr{R}_0}_{\tau} = \displaystyle\frac{-\tau}{(\delta + \tau + \mu)} < 0.
\end{align*}

We use sensitivity analysis to explain how our reproduction number should be interpreted. Parameters with negative sensitivity indices $(\mu, \pi, \delta, \tau )$ have the effect of lowering the community's TB burden if their values increases (i.e., the disease's reproduction number decreases as its parameter values increase). Additionally, if the values of those parameters with positive sensitivity indices $(\Lambda, \beta, \gamma, \theta)$ increase, they play a large part in the spread of tuberculosis in the community. This means that if the values of those parameters increase, so does secondary infection in the community. Therefore, it is essential to increase the negative indices and decrease the positive indices to reduce the disease incidence in the community. Since increasing the human mortality rate to contain disease epidemics is unethical, we do not take this into consideration when studying sensitivity. 

However, many of these studies have either relied on data from countries with more robust health infrastructures or have focused on theoretical applications rather than incorporating region-specific epidemiological data. 

\section{Parameter Estimation}\label{sec_seven}
The SVEITRS model is an extension of the SEIR model, that is specifically tailored to tuberculosis dynamics, and it has several compartments that represent different stages of the disease and different types of interventions. To accurately reflect and capture the transitions of individuals between different disease states, incorporating region-specific epidemiological data is crucial during model parameter estimation. A significant gap in the current literature is the use of locally relevant temporal data to estimate model parameters accurately, with particular emphasis on the choice of estimation techniques. Hence, to accomplish these goals, we use a nonlinear least squares method by implementing minimum search algorithms (gradient descent) and maximum likelihood estimation (MLE) to estimate parameters. We employ these methods to estimate relevant parameters for our model on a synthetic example case and assess the performance and validity of our model via real tuberculosis data from Ethiopia from 2011-2021.

The systems of ODEs in Equation \ref{eq:4.1} can be expressed in more general form as:
\begin{equation}\label{par.eq}
\Dot{\mathbf{x}}(t) = f(\mathbf{x}(t),\theta), ~~ \mathbf{x}(t_0) = \mathbf{x}_0
\end{equation}
where $\mathbf{x}(t) = \big[S(t), V(t), E(t), I(t), T(t), R(t)\big]^t$ and $\mathbf{x}_0 = \mathbf{x}(0)$ are 6-dimensional state vectors and constant initial states, respectively; $\boldsymbol{\theta} = \big[\varLambda, \beta, \gamma, \nu, \pi, \eta, \theta, \delta, \sigma, \tau, \alpha, \omega\big]^t$ is a 12-dimensional vector of the model parameters, and $f : \Omega\subset \mathbb{R}^6\times\mathbb{R}^9\rightarrow\mathbb{R}^6$ is a continuous map, which is described on the open set
\begin{equation}\label{par.eq1}
\Omega = \big\{ (\mathbf{x}(t),\boldsymbol{\theta})\in\mathbb{R}^6\times\mathbb{R}^9|\mathbf{x}_n(t),~\boldsymbol{\theta}_m>0~ \text{for}~ n = 1,\cdots,6~ \text{and}~ m=1,\cdots, 12\big\}.
\end{equation}

The above notation for the SVEITRS models has been used in to illustrate the basic setup of the parameter estimation problem. If we have given the initial conditions $\mathbf{x}_0$ and a set of parameters $\boldsymbol{\theta}$, we can compute the state vectors $\mathbf{x}(t)$ of the model given in Equation \ref{par.eq}. In other words, if we have been given a set of observed data $\mathbf{y}(t), t = 1,\cdots, T$, we try to estimate parameters $\boldsymbol{\theta}$, and then the state vector $\mathbf{x}(t)$ for all $t$ can make predictions for $\mathbf{y}(t)$ for $t > T$. In the following subsection, we describe the two parameter estimation approaches, nonlinear least squares and maximum likelihood estimation, that we use in this study.

\subsection{Nonlinear least squares}
The nonlinear least squares method (NLS) is used to numerically approximate a solution of the model Equation \ref{par.eq1} and observations within the least-squares fitting, where we look for the $\hat{\boldsymbol{\theta}}$ value of the model parameter $\boldsymbol{\theta}$, which minimizes the least squares criterion
\begin{equation}\label{par.eq2}
J(\boldsymbol{\theta}) = \sum_{t=1}^{T} ||\mathbf{y}(t) - f(\mathbf{x}(t),\boldsymbol{\theta})||^2,
\end{equation}
where $f(\cdot)$ is the output of the model for a given set of parameters and where $\boldsymbol{\theta}$ and $\mathbf{y(t)}$ are the observable outputs. The function $J(\boldsymbol{\theta})$ measures the difference between observed data and model predictions, typically using the sum of squared errors. This problem can be classified as a nonlinear least-squares problem because the relationship between the solution system and the parameter $\boldsymbol{\theta}$ is based on a complex system of nonlinear differential equations.

\subsection{Maximum likelihood estimation}
Maximum likelihood estimation (MLE) is a statistical method used to estimate the parameters of a model by maximizing the likelihood function. This technique finds the parameter values that make the observed data most likely, offering a robust approach for parameter estimation in complex models. It is particularly useful in cases where the data are subject to significant variability. The likelihood function is
\begin{equation}\label{par.eq3}
L(\boldsymbol{\theta}) = p(\mathbf{y}|\boldsymbol{\theta})
\end{equation}

For a given set of observations $\mathbf{y}$, the likelihood function is defined as the joint probability of the observed data, given the parameters. This means that we assume that there is some uncertainty in the observed data $\mathbf{y}(t)$ and that the only thing we know about those errors is that they are centered (zero mean) and identically and independently distributed with constant variance; then the likelihood function for the SVEITRS model can be expressed as
\begin{equation}\label{par.eq4}
p(\mathbf{y}|\boldsymbol{\theta}) \propto \exp\left(-\frac{1}{2\sigma_y^2}J(\boldsymbol{\theta})\right) 
\end{equation}
where $\sigma^2_y$ is a fixed variance that can be estimated together with the parameters. The log-likelihood function, $\log L(\boldsymbol{\theta})$, is often used instead of the likelihood function because it simplifies the computations. Maximizing the log-likelihood function is equivalent to minimizing the negative log-likelihood function, and then the optimization problem for MLE can be formulated as:
\begin{equation}\label{par.eq5}
\hat{\boldsymbol{\theta}} = \arg\max_{\boldsymbol{\theta}}\left( p(\mathbf{y}|\boldsymbol{\theta})\right) = \arg\min_{\boldsymbol{\theta}}\left( -\ln{p(\mathbf{y}|\boldsymbol{\theta})}\right)
\end{equation}

\subsection{Synthetic Example Case}
To illustrate the parameter estimation process and check the model's predictive performance and model validity, let us consider a synthetic example case where we generated from the model using specific parameter values and initial conditions. $\beta = 0.05, \gamma = 0.48, \varLambda = 15, \eta = 0.2, \theta = 0.5, \pi = 0.1, \mu = 0.0157, \delta = 0.67, \tau = 0.275, \alpha = 0.88, \omega = 0.12, S(0) = 861, V(0) = 100, E(0) = 20, I(0) = 10, T(0) = 6$, and $R(0) = 3$ in a time range of $t(0) = 0$ to $t(n) = 50$. Then, we estimate the contact rate $(\beta)$ and transmission rate $(\gamma)$ from the generated data, given the other parameters and the generated data, via the minimum search algorithm and maximum likelihood estimation techniques.

The estimated parameter values for $\beta$, $\gamma$, and $\sigma^2_y$ are given in Table \ref{tab:my_label}. Through the parameter estimation techniques, we were able to return to the initial parameter values that we used to make the synthetic data, as shown in the table. The error for both parameters was very small. This shows the performance of the parameter estimation techniques. A comparison of the performance of the minimum search algorithm with that of MLE, reveals that the MLE is more powerful for the current model. Furthermore, we also estimate the 97.5\% confidence intervals of $\beta$ to be $(0.04998402, 0.04999678)$, $\gamma$ to be $(0.47997842, 0.48001454)$, and $\sigma_y^2$ to be $(0.01790117, 0.01790117)$ when MLE is used.
\begin{table}[H]
\centering
\begin{adjustbox}{width=\textwidth}
\begin{tabular}{|l|l|l|l|l|l|}
\hline
Parameters & True & Minimum & Error magnitude & MLE & Error magnitude\\
 & values & search (MS) & (MS)  &  & (MLE) \\
\hline
$\beta$ & 0.05 & 0.04789474 & 0.00210526 & 0.04999040 & 0.0000096\\
\hline
$\gamma$ & 0.48 & 0.4789474 & 0.0010526 & 0.47999648 & 0.00000352\\
\hline
$\sigma_y^2$ &  &  &  & 0.01790117 & \\
\hline
\end{tabular}
\end{adjustbox}
\caption{Estimated parameters of the contact rate $\beta$, infection rate $\gamma$, magnitudes of errors, and fixed common standard deviation ($\sigma_y^2$) for the synthetic data.}
\label{tab:my_label}
\end{table} 


\begin{figure}[H]
\centering
\subcaptionbox{Synthetic data.}
{\includegraphics[width=0.45\textwidth]{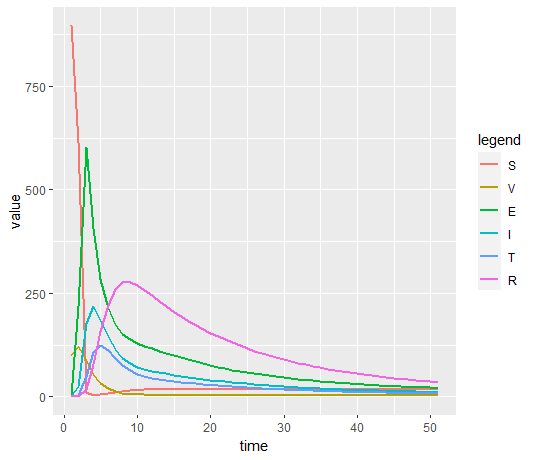}}
\subcaptionbox{Model output}
{\includegraphics[width=0.45\textwidth]{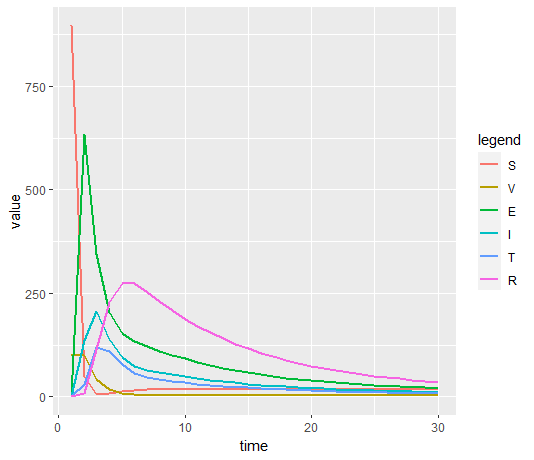}}  
\caption{A synthetic SVEITRS model simulated solution. The result is obtained with the following parameters: $\beta = 0.05, \gamma = 0.48, \varLambda = 15, \eta = 0.2, \theta = 0.5, \pi = 0.1, \mu = 0.0157, \delta = 0.67, \tau = 0.275, \alpha = 0.88, \omega = 0.12$, and initial values: $N = 1000, S(0) = 861, V(0) = 100, E(0) = 20, I(0) = 10, T(0) = 6, R(0) = 3$, years = 50.
}
\label{fig:enter-label10}
\end{figure}
We also checked how well the proposed model worked with the estimated parameters by comparing the fake data for each compartment with the model's results via the estimated parameters, as shown in the subfigures of Figure \ref{fig:enter-label10}. The results in the figures for all the compartments indicate the reliability and performance of the proposed model with the corresponding parameter estimation. For the real dataset, we replace the true generated synthetic data with actual observed data of Ethiopian tuberculosis cases from 2011-2021 and implement parameter estimation for the real data case. We present the parameter values obtained and initial conditions in Table \ref{tab:6.2} for the real Ethiopian tuberculosis data case from 2011-2021.

\begin{table}[H]
\centering
\begin{adjustbox}{width=\textwidth}
\begin{tabular}{clll}
\hline
\textbf{Parameter} & \textbf{Description} & \textbf{Value} & \textbf{Source}  \\
\hline
$N_0$ & Initial number of total population &$9.18\times10^7$ & \cite{worldometers2024}  \\
\hline
$S_0$ & Initial number of susceptible individuals & $7.505\times 10^7$ & Estimated \\
\hline
$V_0$ & Initial number of vaccinated individuals & $1.205\times10^6$ & Estimated\\
\hline
$E_0$ & Initial number of exposed individuals & $1.503\times10^7$& Estimated \\
\hline
$I_0$ & Initial number of infectious individuals & $2.323\times10^5$ & \cite{worldbank_tb_ethiopia} \\
\hline
$T_0$ & Initial number of treated individuals & $1.56\times10^5$ & \cite{tbdiah} \\
\hline
$R_0$ & Initial number of recovery individuals & $1.39\times10^5$ & \cite{worldbank_tuberculosis}\\
\hline
$\Lambda$ & Recruitment rate & $1.658\times10^6$ & Estimated \\
\hline
$\beta$ & The contact rate & $7.75\times 10^{-7}$ & Estimated \\
\hline
$\gamma$ & Infection rate & $0.225 $ & Estimated \\
\hline
$\mu$ & Natural death rate & $0.0157$ & Estimated \\
\hline
$\pi$ & Vaccination rate & $0.713$ & \cite{tradingeconomics2023} \\
\hline
$\eta$ & loss rate of immunity by recovered one & $0.2$ & Assumed \\
\hline
$\theta$ & Loss of protection for vaccination & $0.5$ & \cite{colditz1994efficacy} \\
\hline
$\delta$ & treatment rate of I(t) & $0.67$ & \cite{tbdiah}  \\
\hline
$\tau$ & disease-causing death rate of I(t) & $0.266$ & \cite{tradingeconomics} \\
\hline
$\alpha$ & Recovery rate for treated individuals & $0.89$ & \cite{worldbank_tuberculosis} \\
\hline
$\omega$ & disease-causing death rate of T(t) & $0.11$ & \cite{worldbank_tuberculosis} \\ 
\hline
\end{tabular}
\end{adjustbox}
\caption{Estimated values of parameters for the model}
\label{tab:6.2}
\end{table}

\section{Results and Discussions}\label{sec_eight}
In this section, we discuss the model outputs obtained and the corresponding parameter estimations for real Ethiopian tuberculosis data from 2011-2021. The numerical simulation results are obtained for the dynamical system of the tuberculosis disease model equation \(\left( \ref{eq:4.1} \right)\), and corresponding computations and graphical illustrations are carried out via Matlab and R programming languages. The idea behind numerical simulations is to look into how the proposed mathematical model works and why it makes sense. This model is too complicated to solve analytically, similar to most nonlinear systems. The simulations are performed by using the set of parameter values given in Table \ref{tab:6.2} above.
   
The main goal of this study is to find the contact rate ($\beta$) and transmission rate ($\gamma$) for tuberculosis (TB) by existing TB incidence data, along with other computed parameter values, and to investigate disease changes over time and how well the proposed model works. The comparison of the model output and the incidence data for the estimated parameter values is given in Figure \ref{fig:enter-label11} with a 95\% confidence level. The figure shows that the suggested SVEITRS model can explain how the disease changes over time. The estimated parameter values are very good, so the model's output closely matches the real data on disease incidence.  
\begin{figure}[H]
\centering
\includegraphics[width=1\linewidth]{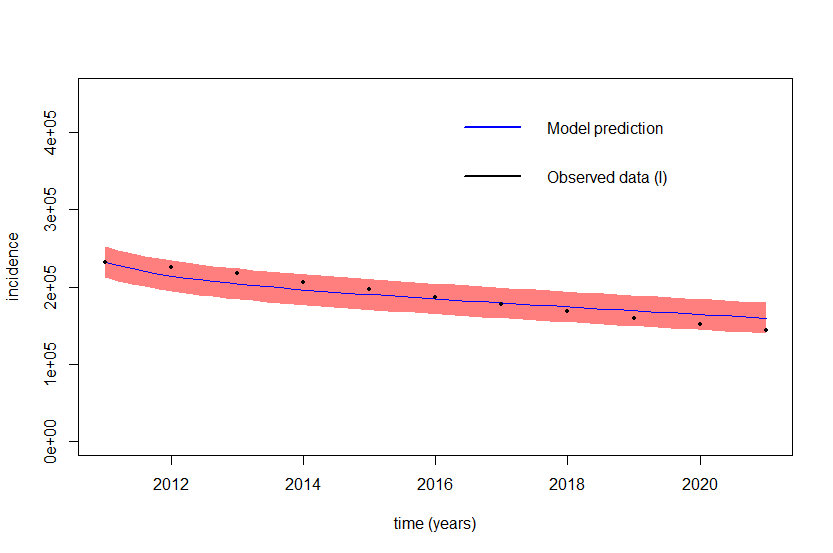}
\caption{The fitted model of the incidence cases to the reported cases is used for the model for Ethiopia from 2011-2021.}
\label{fig:enter-label11}
\end{figure}
In Figure \(\left( \ref{fig:SVEITRS_model} \right)\), we present the model outputs of all the compartments to illustrate how the different classes evolve over time under the estimated model parameters with initial conditions as the data in 2011. According to Figure \(\left( \ref{fig:SVEITRS_model} \right)\), as time increases, the susceptible human populations decrease rapidly, whereas the vaccinated population start to increase initially and then decrease rapidly.
\begin{figure}[H]
\centering
\includegraphics[width=1\linewidth]{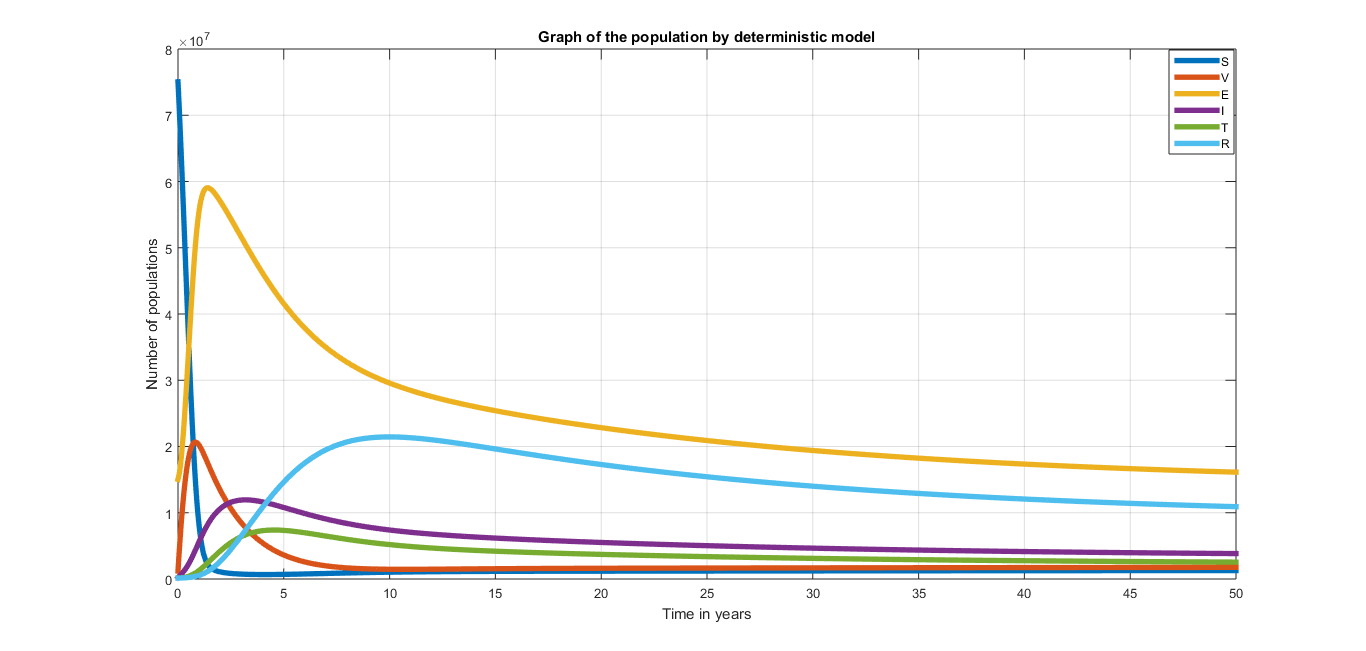}
\caption{A plot indicating model outputs and how each class evolves in time for the estimated parameters.}
\label{fig:SVEITRS_model}
\end{figure}
This finding indicates that susceptible human and vaccinated population will continue to join the exposed class of the disease; as a result, the exposed populations initially tends to increase, reflecting the growth of the disease in the population. However, over time, the exposed population starts to decline, likely because of the flow into the infected class, the treatment, and the recovery processes. We obtained $\mathscr{R}_0 = 2.67$ for the estimated model parameters, which are greater than one, indicating that the tuberculosis is continuously multiplying. 

Moreover, we have checked the effects of contact rate $(\beta)$ and infectious rate $(\gamma)$ on the exposed and infected human population, as shown in Figures \ref{fig:beta_exposed},\ref{fig:beta_infected},\ref{fig:gamma_exposed} and \ref{fig:gamma_infected} by varying $(\beta)$ and $(\gamma)$ starting from their estimated values.

\subsection{Effect of contact rate $(\beta)$ on the TB-exposed population}
According to Figure \ref{fig:beta_exposed}, as the contact rate $(\beta)$ increases, the exposed population also increases. This is because when the contact rate increases, the reproduction number also increases, and the disease persists in the population. This shows that the contact rate between susceptible and infected individuals increase. This graph also demonstrates that the contact rate has a large effect on the spread of the disease through the population. If the contact rate is high, then the rate of infection of the disease will also be high, as would be expected logically. Hence, to decrease the number of secondary infections, one can decrease the contact rate between susceptible and infected individuals. 
\begin{figure}[H]
\centering
\includegraphics[width=1.0\linewidth]{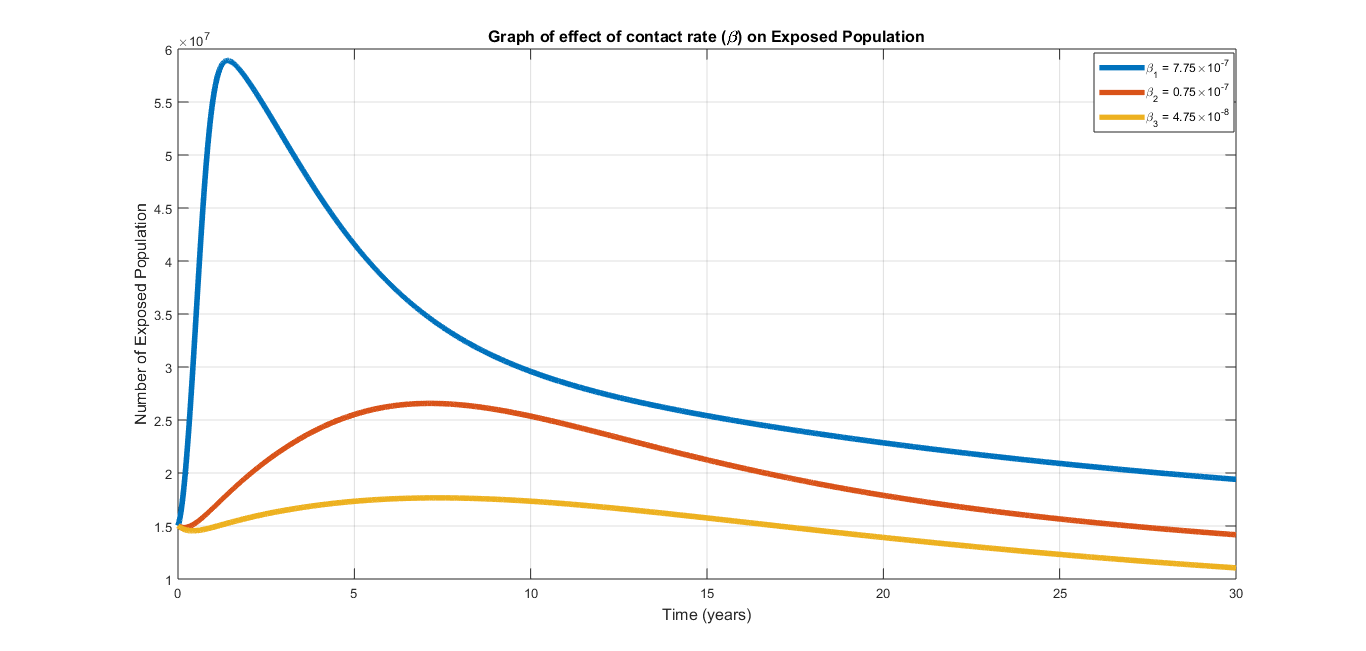}
\caption{The graphical result of the effect of the contact rate $(\beta)$ on the exposed population for several values of the contact rate $(\beta)$.}
\label{fig:beta_exposed}
\end{figure}
\subsection{Effect of contact rate $(\beta)$ on the TB-infected population}
In Figure \ref{fig:beta_infected}, we show the effects of the contact rate $(\beta)$ on the number of infected people. We obtained the numerical results by varying the value of the contact rate $(\beta)$ while keeping the other parameters constant. When the value of the contact rate $(\beta)$ increases from $4.75\times10^{-8}$ to $7.75\times10^{-7}$ there is an increase in the infected population. We see that as the contact rate increases, the number of infected people in the community increases as expected. Higher contact rates lead to more frequent interactions between susceptible and infected individuals, causing the disease to spread faster and increasing the number of infected people quickly. 
\begin{figure}[H]
\centering
\includegraphics[width=1.0\linewidth]{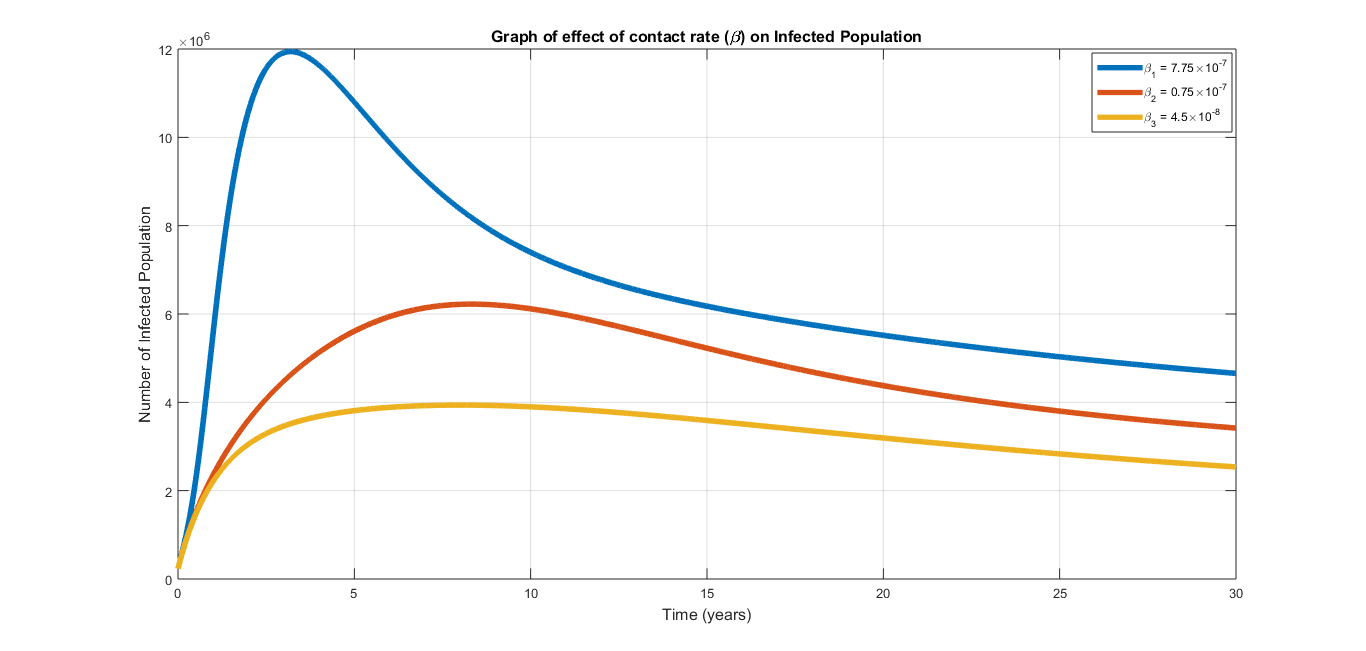}
\caption{The graphical result of the effect of contact rate $(\beta)$ on the infected population for several values of the infectious rate $(\beta)$.}
\label{fig:beta_infected}
\end{figure}
\subsection{Effect of infection rate $(\gamma)$ on the TB-exposed population} 
In Figure \ref{fig:gamma_exposed}, we show the effects of the infection rate $(\gamma)$ on the number of exposed populations. We obtained the numerical results by varying the value of the infectious rate $(\gamma)$ while keeping the other parameters constant. When the value of the infection rate $(\gamma)$ increased from 0.075 to 0.227, there was a decrease in the exposed population. This is because an increase in the infection rate results in an increase in the reproduction number; that is, the disease persists in the population. This shows the increase in infection rate between exposed and infected individuals. Hence, to decrease the number of secondary infections, one can decrease the infection rate between exposed and infected individuals.
\begin{figure}[H]
\centering
\includegraphics[width=1.0\linewidth]{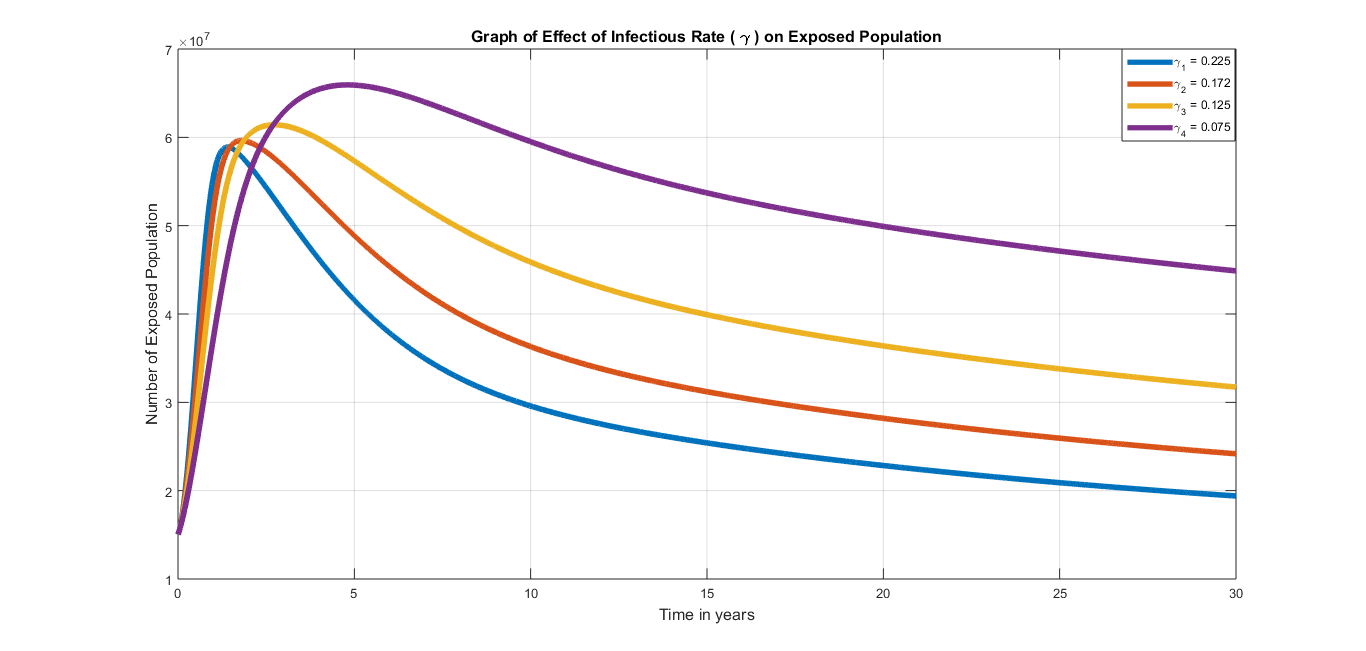}
\caption{The graphical result of the effect of the infection rate $(\gamma)$ on the exposed population for several values of the infectious rate $(\gamma)$.}
\label{fig:gamma_exposed}
\end{figure}
\subsection{Effect of the infection rate $(\gamma)$ on the TB-infected population} 
In Figure \ref{fig:gamma_infected}, we demonstrate the impact of the infection rate $\gamma$ on the number of infected individuals. The numerical results were obtained by varying the value of the infectious rate $\gamma$ while keeping the other parameters constant. In the deterministic model, when the value of the infection rate increased from 0.075 to 0.227, there was also an increase in the number of infected individuals as in the exposed cases. 
\begin{figure}[H]
\centering
\includegraphics[width=1.0\linewidth]{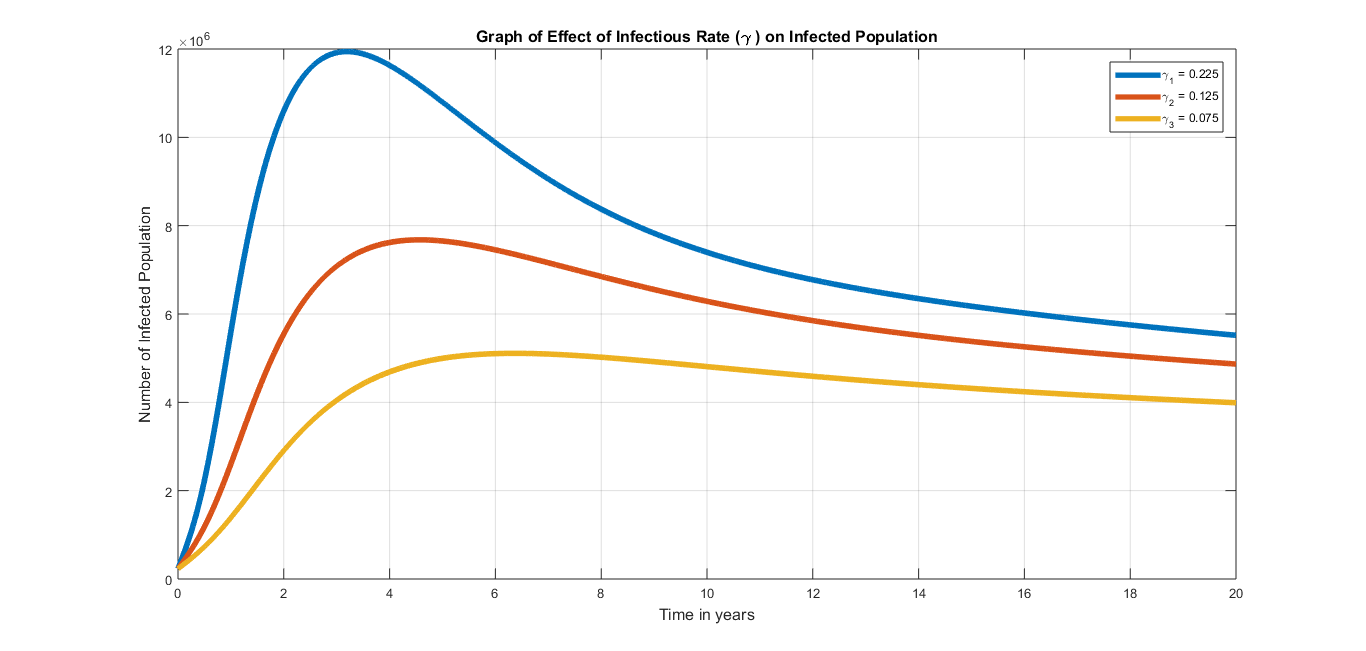}
\caption{The graphical result of the effect of the infectious rate $(\gamma)$ on the infected population for several values of the infectious rate $(\gamma)$.}
\label{fig:gamma_infected}
\end{figure}
In this case again, the increase in the infection rate results in increase in the reproduction number; that is, the disease persists in the population. Hence, to decrease the number of secondary infections, one can decrease the infection rate between latent and actively infected individuals as noted out above.

\subsection{Effect of the treatment rate $\delta$ on the TB-infected population}
In Figure \ref{fig:delta_infected}, we checked the effect of the treatment rate $\delta$ on the infected by varying the value of $\delta$ from 0.47 to 0.67 starting from its estimated value and keeping the rest of the parameters constants. The result revealed that the infectious population decreases as the treatment rate of infected individuals increases; in the long run, this may lead to the eradication of the disease from the community. Thus, we can conclude that increasing the treatment rate $\delta$, can play a vital role in reducing the incidence of tuberculosis from the community locally.
\begin{figure}[H]
\centering
\includegraphics[width=1.0\linewidth]{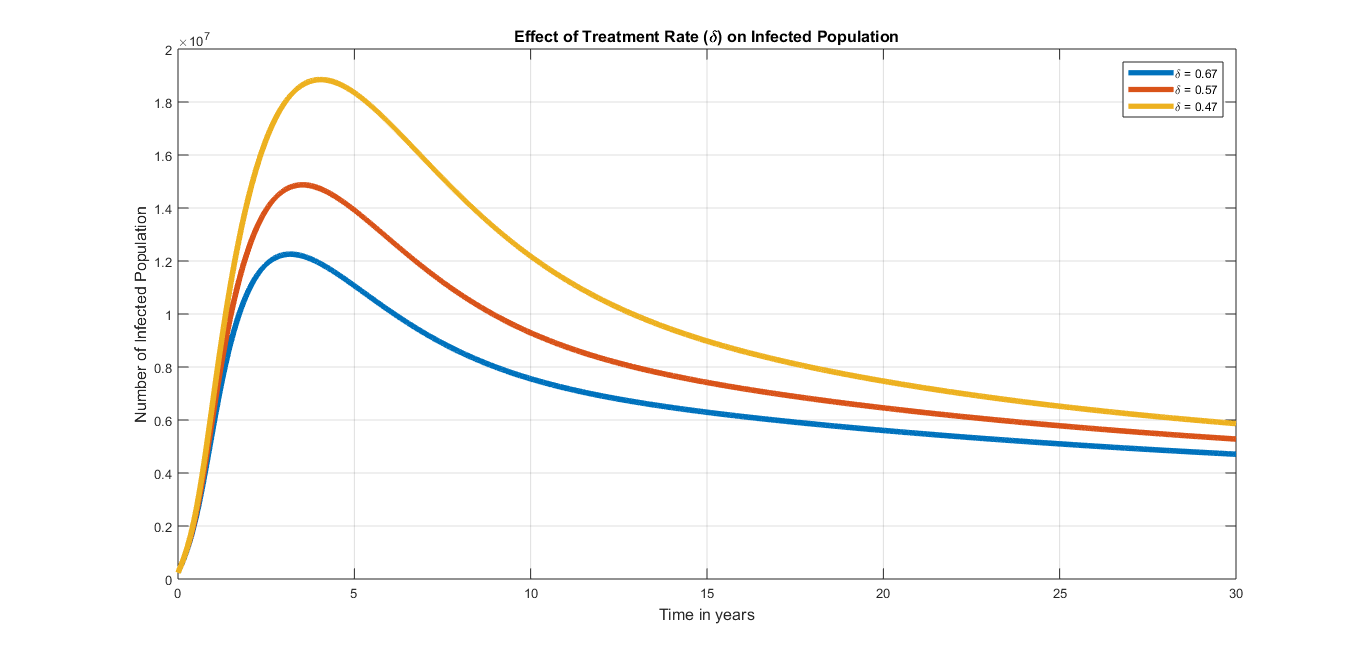}
\caption{The graphical result of the effect of treatment rate $(\delta)$ on the infected population for several values of the infectious rate $(\delta)$.}
\label{fig:delta_infected}
\end{figure}

\section{Conclusions} \label{sec_nine}
In this study, we model tuberculosis transmission dynamics in Ethiopia by incorporating imperfect vaccination and treatment for active TB-infected classes. We estimated the contact and infection rate parameters through minimum search algorithms and maximum likelihood estimation techniques using Ethiopian tuberculosis data from 2011-2021. The model output is illustrated via the parameters estimated from the Ethiopian data. The model analysis revealed the existence of a unique, bounded, and positively invariant solution. We find the disease-free equilibrium (DFE) and endemic equilibrium (EE) states and use the next-generation matrix operator to find the basic reproduction number $\mathscr{R}_0$. The analytical results show that the DFEs are locally and globally asymptotically stable when $\mathscr{R}_V<1$, whereas the EEs are locally and globally asymptotically stable if $\mathscr{R}_V>1$. This study revealed that contact rate ($\beta$), recruitment rate ($\Lambda$), the loss of protection for the vaccination rate ($\theta$), and rate of progression from the exposed class to the infectious class ($\gamma$) were the most important factors. These parameters have the greatest impact on increasing the reproductive number, significantly affecting TB spread in Ethiopia. Conversely, parameters $\pi, \mu, \delta$, and $\tau$ have the greatest influence on decreasing the reproductive number, leading to tuberculosis elimination.

Furthermore, parameter estimation for the contact and infection rates was performed from the incidence data, providing a more accurate parameter values specific to Ethiopia for the proposed TB dynamics model. Overall, the numerical simulation results, as discussed in the previous sections, suggest that treating infectious individuals significantly affects the eradication of TB transmission for active TB patients locally. We conclude that increasing the contact rate leads to an increase in the TB-exposed population. Similarly, increasing the infection rate has led to an increase in the TB-infected population. The TB-infected population may become large, and treating such a large population is also a challenge economically for developing countries such as Ethiopia. Therefore, treating such a large number of patients may lead to an increase in the number of drug-resistant tuberculosis cases. Therefore, focusing on controlling contact and infection rate parameters is crucial, in addition to considering educational and awareness programs to reduce the transmission of TB bacteria from infected to susceptible populations.

By implementing this multifaceted approach that combines reduced contact rates and increased vaccination coverage, policymakers and healthcare workers can develop targeted and impactful interventions to effectively control and eventually eliminate TB in Ethiopia and other high-burden settings. 

\section*{ACKNOWLEDGMENTS}
Yassin Tesfaw Abebe acknowledges support from the International Mathematical Union (IMU) and the Graduate Research Assistantships in Developing Countries (GRAID) Program.

\section*{AUTHOR DECLARATIONS}
\subsection*{Conflict of interest}
No potential conflicts of interest were reported by the author(s).

\subsection*{Author Contributions}
\textbf{Moksina Seid:} Conceptualization (equal); investigation (equal); writing – original draft (equal); writing – review \& editing (equal); software (equal).  \textbf{Yassin Tesfaw Abebe:} Writing – original draft (equal); Software (equal); writing – review \& editing (equal). \textbf{Abdu Mohammed Seid}: Conceptualization (equal); investigation (equal); and supervision.

\section*{DATA AVAILABILITY}
The data used in this paper are accessed online from \cite{worldbank_tb_ethiopia, worldbank2023, worldometers2024, tbdiah, tradingeconomics2023, worldbank_tuberculosis, tradingeconomics}.

\bibliography{refs}

\end{document}